\documentclass{article}
\usepackage{a4wide}
\usepackage[applemac]{inputenc}
\usepackage[english]{babel}
\usepackage{latexsym}
\usepackage{amssymb}
\usepackage{amsmath}
\usepackage{amsthm}
\usepackage{cite}    
\usepackage{ifthen,xspace}

\usepackage{tikz}
\usetikzlibrary{decorations.pathmorphing,decorations.pathreplacing,automata}
\usepackage{enumerate}

\newcommand{\set}[2]{\left\{#1\mathrel{\left|\vphantom{#1}\vphantom{#2}\right.}#2\right\}}
\newcommand{\oneset}[1]{\left\{\mathinner{#1}\right\}}

\newcommand{\abs}[1]{\left|\mathinner{#1}\right|}

\newcommand{\N}{\mathbb{N}}

\newcommand{\Oh}{\mathcal{O}}

\newcommand{\cA}{\mathcal{A}}

\newcommand{\cF}{\mathcal{F}}

\newcommand{\cH}{\mathcal{H}}
\newcommand{\cI}{\mathcal{I}}

\newcommand{\cQ}{\mathcal{Q}}

\newcommand{\BAR}{\overline{\phantom{ii}}}

\newcommand{\smalloverline}[1]
{{\mspace{1mu}\overline{\mspace{-1mu}#1\mspace{-1mu}}\mspace{1mu}}}

\newcommand{\ov}[1]{\smalloverline{#1}}

\newcommand{\refthm}[1]{Theorem~\ref{#1}}

\newcommand{\reflem}[1]{Lemma~\ref{#1}}
\newcommand{\refprop}[1]{Proposition~\ref{#1}}
\newcommand{\refrem}[1]{Remark~\ref{#1}}

\newcommand{\reffig}[1]{Figure~\ref{#1}}
\newcommand{\refsec}[1]{Section~\ref{#1}}

\newcommand{\hpc}{hairpin completion\xspace}

\newcommand{\lcf}{linear con\-text-free\xspace}

\newcommand{\nd}{non-deterministic\xspace}

\newcommand{\NL}{{\bf NL}\xspace}

\newcommand{\TM}{{T}uring machine\xspace}
\newcommand{\svlstd}{single-valued \nd $\log$-space transduction\xspace}
\newcommand{\Svlstd}{Single-valued \nd $\log$-space transduction\xspace}

\newcommand{\e}{1}

\newcommand\lds{,\ldots ,} 
 
\newcommand\ccH{\ensuremath{\cH_k(L_1,L_2)}\xspace}

\newcommand\Hk{\cH_k}

\newcommand{\sse}{\subseteq}
\newcommand{\es}{\emptyset}

\newcommand{\dead}{t}

\renewcommand{\phi}{\varphi}

\newcommand{\alp}{\alpha}
\newcommand{\bet}{\beta}
\newcommand{\gam}{\gamma}
\newcommand{\del}{\delta}

\newcommand{\Sig}{\Sigma}

\newcommand{\gabag}{\gamma\alpha\beta\ov\alpha\ov\gamma}
\newcommand{\gaba}{\gamma\alpha\beta\ov\alpha}
\newcommand{\abag}{\alpha\beta\ov\alpha\ov\gamma}

\newcommand\RAS[1]{\overset{#1}\Longrightarrow}
\newcommand\ras[1]{\overset{#1}\longrightarrow}

\theoremstyle{plain}
\newtheorem{theorem}{Theorem}[section]
\newtheorem{proposition}[theorem]{Proposition}
\newtheorem{lemma}[theorem]{Lemma}
\newtheorem{corollary}[theorem]{Corollary}

\theoremstyle{definition}

\newtheorem{example}[theorem]{Example}

\theoremstyle{remark}
\newtheorem{remark}[theorem]{Remark}

\newenvironment{test}[1]
{\begin{trivlist}\item[\hskip\labelsep {\bfseries Test #1:\,}]\it}
{\end{trivlist}}

\newenvironment{vd}{\noindent\color{blue} VD }{}

\newenvironment{sk}{\noindent\color{red} SK }{}

\begin{document}

\title{It Is NL-complete to Decide Whether a Hairpin Completion of Regular Languages
Is Regular}

\author 
	{Volker Diekert, Steffen Kopecki \\
	\small
	{\tt \{diekert,kopecki\}@fmi.uni-stuttgart.de} \\
	\small
	University of Stuttgart,
	Institute for Formal Methods in Computer Science (FMI), \\
	\small
	Universit\"atsstra\ss e 38,
	D-70569 Stuttgart}

\maketitle

\begin{abstract}
	The hairpin completion is an operation on formal languages
	which is inspired by the hairpin formation in biochemistry.
	Hairpin formations occur naturally within DNA-computing. 
	It has been known that the hairpin completion of a regular language
	is linear context-free, but  not regular, in general.
	However, for some time it is was open whether the regularity of the
	hairpin completion of a regular language is is decidable.
	In 2009  this decidability problem has been solved positively in \cite{DiekertKM09} by 
	providing a polynomial time algorithm.
	In this paper we improve the complexity bound by showing
	that the decision problem is actually \NL-complete. 
	This complexity bound holds for both, the one-sided and the two-sided 
	\hpc{}s. 

	\vspace{\baselineskip}
	\NL-Complete Problems;
	DNA-Computing; Hairpin Completion.
\end{abstract}



\section{Introduction}

The hairpin completion is a  natural operation of formal languages 
which has been inspired 
by molecular  phenomena in biology and which occurs naturally during  DNA-compu\-ting. 
An intramolecular base pairing, known as a \emph{hairpin},
is a pattern that can occur in single-stranded DNA and, more commonly, in RNA.
Hairpin or hairpin-free structures have numerous applications to DNA computing and
molecular genetics, 
see \cite{garzon1,garzon2,garzon3,KariKLST05,KariMT07}
and the references within for a detailed discussion.
For example, an instance of \textsc{3-Sat} has been solved with a
DNA-algorithm and one of the main concepts was to eliminate
all molecules with a hairpin structure, see \cite{KensakuSakamoto05192000}.

In this paper we study the \hpc from a purely formal language viewpoint.
The \hpc of a formal language was first defined by Cheptea, Mart{\'\i}n-Vide, and Mitrana in \cite{ChepteaMM06};
here we use a slightly more general definition which was introduced in \cite{DiekertKM09}.
The hairpin completion and some related operations have been studied in a series of papers
from language theoretic and algorithmic point of view, see e.g.,
\cite{DBLP:conf/cie/ManeaM07,ManeaMY09tcs,ManeaMM09,Ito2010,ManeaMY10,ManeaMM10,Kopecki10}.
The formal operation of the \hpc on words is best explained in Figure~\ref{whatanicehairpin}. 
In that picture as in the rest of the 
paper we mean by putting a \emph{bar} on a word (like $\ov \alp$) to read it
from right-to-left and in addition to replace a letter $a$ with 
the (Watson-Crick) complement $\ov a$. 
The hairpin completion of a regular language is linear context-free, but  not regular, in general \cite{ChepteaMM06}. 
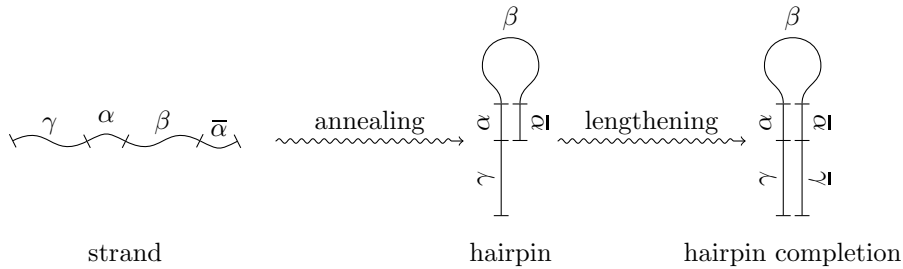
\begin{figure}[h]
  \centering
 
  \begin{tikzpicture}[text height=1.5ex,text depth=.25ex]
    \draw [|-|] (-4,1) .. controls +(.5,.25) and +(-.5,-.25) ..
      node [above] {$\gamma$} (-3,1);
    \draw [-] (-3,1) .. controls +(.25,.125) and +(-.25,.125) .. 
      node [above] {$\alpha$} (-2.5,1);
    \draw [|-|] (-2.5,1) .. controls +(.5,-.25) and +(-.5,.25) ..
      node [above] {$\beta$} (-1.5,1);
    \draw [-|] (-1.5,1) .. controls +(.25,-.125) and +(-.25,-.125) .. 
      node [above] {$\ov\alpha$} (-1,1);

	\draw [->,decorate,decoration=
			{snake,amplitude=.3mm,segment length=1.5mm,post length=1mm}]
		(-.5,1) -- (2,1) node [above,text centered,midway] {annealing};
    \draw [|-] (2.5,0) -- node [above,sloped] {$\gamma$} (2.5,1);
    \draw [|-|] (2.5,1) -- node [above,sloped] {$\alpha$} (2.5,1.5);
    \draw [-] (2.5,1.5) .. controls +(up:.25) and +(down:.25) .. (2.25,2)
      .. controls +(up:.5) and +(up:.5) .. node [above,sloped] {$\beta$} (3,2)
      .. controls +(down:.25) and +(up:.25) .. (2.75,1.5);
    \draw [|-|] (2.75,1.5) -- node [above,sloped] {$\ov\alpha$} (2.75,1);

	\draw [->,decorate,decoration=
			{snake,amplitude=.3mm,segment length=1.5mm,post length=1mm}]
		(3.25,1) -- (5.75,1) node [above,text centered,midway] {lengthening};

    \draw [|-] (6.25,0) -- node [above,sloped] {$\gamma$} (6.25,1);
    \draw [|-|] (6.25,1) -- node [above,sloped] {$\alpha$} (6.25,1.5);
    \draw [-] (6.25,1.5) .. controls +(up:.25) and +(down:.25) .. (6,2)
      .. controls +(up:.5) and +(up:.5) .. node [above,sloped] {$\beta$}
      (6.75,2) .. controls +(down:.25) and +(up:.25) .. (6.5,1.5);
    \draw [|-|] (6.5,1.5) -- node [above,sloped] {$\ov\alpha$} (6.5,1);
    \draw [-|] (6.5,1) -- node [above,sloped] {$\ov\gamma$} (6.5,0);

    \node at (-2.5,-.5) {strand};
    \node at (2.625,-.5) {hairpin};
    \node at (6.375,-.5) {hairpin completion};
  \end{tikzpicture}

  \caption{Hairpin completion of a DNA-strand (or a word).}
\label{whatanicehairpin}
\end{figure}

For some time it was not known whether regularity of the hairpin completion
 of a regular language
is decidable.
It was only in 2009 when  we presented in \cite{DiekertKM09} a decision algorithm.
Actually, we proved a better result by providing a polynomial time algorithm
with a (rough) runtime estimation of about $\Oh(n^{20})$.

In an extended abstract which appeared at the CIAA~2010
we presented a modified  approach to solve the same problem \cite{DieKop11}
which led, in particular, to the following two new results:  
First, the time complexity of  the new 
decision algorithm is in  $\Oh(n^8)$. Second,  the decision problem is 
NLOGSPACE-complete, i.e., \NL-complete. 

This paper is the journal version  of \cite{DieKop11} for the second result.  We decided to focus on the space complexity since, in terms of complexity,  \NL-completeness yields a precise characterization and because  
 the given page limit did not allow to include full proofs for 
all results of \cite{DieKop11}. 
Moreover, our proofs are still rather 
technical and the focus on the \NL-algorithm simplifies the presentation. 

We consider the  one-sided and the two-sided 
	\hpc{}s simultaneously.
	It turns out that \NL-completeness holds in both cases. 
	
	The paper is organized as follows. In \refsec{secpn} we fix the 
	notation used throughout. We give the formal definition of the \hpc $\ccH$
	and we discuss our input model using appropriate deterministic automata. 
	
	In \refsec{secmain} we state the main result (\refthm{thm:main}) and we give 
	a full proof in the subsequent subsections. 
	A main technical tool is the use of  \svlstd{}s, which might be not fairly standard. They are 
	explained in \refsec{svlstd}. In \refsec{secopen} we give a short conclusion and 
	we state some open problems.

\section{Preliminaries and Notation}\label{secpn}
We assume the reader to be familiar with the basic concepts of formal
language theory, automata theory, and complexity theory, as one can find in the text books\cite{HU,pap94}. 
By \NL we mean the complexity class NLOGSPACE, which contains the 
problems which can be decided by  a \nd{} \TM  using 
$\Oh(\log n)$ work space. Throughout we use the well-known result that 
\NL is closed under complementation, see e.g. \cite{pap94}. We also use the fact that 
if $L$ can be reduced to $L'$ via  some \svlstd and $L' \in $ \NL, then we have 
 $L \in $ \NL, see \cite{AlvarezJenner95} and \refsec{svlstd} for more details. 
 
By $\Sig$ we denote a finite alphabet with at least two letters.
The set of words over $\Sig$ is denoted $\Sig^*$; and the \emph{empty
word} is denoted by $\e$. Given a word $w$, we denote by $|w|$ its length
and $w(m)\in \Sig$ its $m$-th letter.
If $w=xyz$ for some $x,y,z\in \Sig^*$, then $x$ and $z$ are called \emph{prefix} and 
\emph{suffix} of $w$, respectively.
By a proper prefix $x$ of $w$ we mean a prefix such that $x\neq w$ (but we allow $x=\e$).
The prefix relation between words $x$ and $w$ is denoted 
by $x\leq w$ and for proper prefixes by $x < w$.

We assume that the alphabet $\Sig$ is equipped
with an involution  $\BAR: \Sig\to \Sig$.
An \emph{involution} for a set is  a bijection  such that
$\overline{\ov{a}} = a$. 
We extend the involution to words $a_1 \cdots a_n $ 
by  $\ov{a_1 \cdots a_n} =  \ov{a_n} \cdots \ov{a_1}$ where the $a_i$'s are letters. This convention is  like taking inverses in groups.
For languages $L \sse \Sig^*$ we denote by $\ov L$ the set $$\ov L = \set{\ov w}{w\in L }.$$
Throughout the paper $L_1, L_2$ are two regular languages in $\Sig^*$ and
by $k$ we mean a (small) 
constant. (In a biological setting $k \sim 10$ yields a reasonable value.)
According to \reffig{whatanicehairpin} we  define the \emph{hairpin completion} $\ccH$ by 
$$\ccH= \set{\gamma \alpha \beta \ov{\alpha} \ov{\gamma}}{(\gamma \alpha \beta \ov{\alpha}\in L_1
\vee \alpha \beta \ov{\alpha}
 \ov{\gamma} \in L_2 ) \wedge {|\alpha|}=k}.
$$
This definition is slightly more general than the original definition in
\cite{ChepteaMM06,ManeaMY09tcs}.
It allows us to treat the two-sided hairpin completion ($L_1=L_2$) and the 
one-sided hairpin completion (either $L_1=\es$ or $L_2 = \es$) simultaneously. 

A regular language can be specified by a \nd finite automaton (NFA) $\cA= (\cQ, \Sig, E, \cI, \cF)$,
where $\cQ$ is the finite set of \emph{states}, $\cI \sse \cQ$ is the set of 
\emph{initial states}, and  $\cF \sse \cQ$ is the set of 
\emph{final states}. The set $E$ contains labeled  \emph{edges} (or  \emph{arcs}),
it is a subset of $\cQ \times \Sig \times 
\cQ$. For a word $u \in \Sig^*$ we write $p \ras{u}{}q$, if there is a path
{}from state $p$ to $q$ which is labeled by the word $u$. Thus,  the
accepted language becomes
$$L(\cA) = \set{u\in \Sig^*}{\exists p \in \cI\, \exists q\in \cF:\; p \ras{u}{}q}.$$

Later it will be crucial to use also paths which avoid final states. For this 
we introduce a special notation. First remove all arcs $(p,a,q)$ where 
$q\in \cF$ is a final state. Thus, final states do not have incoming arcs anymore.
Let us write 
$p \RAS{u}{}q$, if there is a path {}from state $p$ to $q$ which is labeled by the word $u$ in this new  automaton after removing these arcs.
 Note that for such a 
path $p \RAS{u}{}q$ we allow $p\in \cF$, but on the path we never enter any final state again.

An NFA is called a \emph{deterministic finite automaton} (DFA), if it has exactly one initial state
and for every state $p  \in \cQ$ and every letter $a \in \Sig$ there is exactly one arc $(p,a,q) \in E$.
In particular,
in this paper a DFA  is always \emph{complete.} Thus, we can read every word
to its end.  We also write $p \cdot u = q$, if 
$p \ras{u}{}q$. This yields a (totally defined) function $\cQ \times \Sigma^* \to
\cQ$. (It defines an action of $\Sigma^*$ on $\cQ$
on the right.)

In the following we use a DFA accepting $L_1$ as well as a DFA accepting  $L_2$,
but the DFA for $L_2$ has to work from right-to-left.
Instead of introducing this concept we use a 
DFA (working as usual from left-to-right), which 
accepts  $\ov {L_2}$. This  automaton has the same number of states as 
(and is structurally isomorphic to) a DFA accepting the \emph{reversal language} of $L_2$.

As input we assume that the regular languages $L_1$ and $\ov {L_2}$ are  specified
by  DFAs $\cA_1$ and $\cA_2$  with state set $\cQ_i$, state $q_{0i} \in \cQ_i$ as initial state, and
$\cF_i \subseteq \cQ_i$ as final states.
By $n$ we denote the input size 
$$n = \abs{\cQ_1}+\abs{\cQ_2}.$$ 

We also need  the usual product DFA with 
state space 
$$\cQ = \set{(p_1,p_2)\in  \cQ_1 \times \cQ_2}{\exists w \in \Sig^*:(p_1,p_2)= 
(q_{01}\cdot w,\; q_{02}\cdot w) }  . $$
The action is given 
by $(p_1,p_2)\cdot a = (p_1\cdot a,\; p_2\cdot a).$
As $\cQ$ contains only reachable states, the size of $\cQ$ might be smaller than
$\abs{\cQ_1}\cdot \abs{\cQ_2}$. 
In the following we   work simultaneously in all three automata defined so far. 
Moreover,  in $\cQ_{1}$
and $\cQ_{2}$ we are going to work backwards. This leads to nondeterminism. 

\section{Main result}\label{secmain}

The purpose of this paper is to prove the following result: 
\begin{theorem} \label{thm:main}
	The following problem is
	\NL-complete:
	
{\bf Input:}  Two DFAs $\cA_1$ and $\cA_2$ recognizing $L_1$ and $\ov {L_2}$ with 
state sets $\cQ_1$ and $\cQ_2$ resp.{} such that $n = \abs{\cQ_1}+\abs{\cQ_2}$.

{\bf Question:} Is \ccH regular?
\end{theorem}

Since \NL is included in {\bf P} we obtain the following result {}from 
\cite{DiekertKM09} as a corollary.
\begin{corollary}\label{cor:main}
The problem whether the hairpin completion \ccH is regular is decidable in polynomial time. 
\end{corollary}

We now turn to the proof of \refthm{thm:main}. 
The  \NL-hardness is immediate:

\begin{lemma}
	The problem whether the hairpin completion \ccH is regular is
	\NL-hard, even for $L_2 =\es$. 
\end{lemma}

\begin{proof}
	The well-known \NL-complete {\em Graph-Accessibility-Problem} 
	\cite{pap94} 
	can easily be reduced to the following problem for DFAs:
	
	Let $\Sig = \oneset{a,\ov a,b,\ov b}$ be an alphabet with four letters.
	Decide for a given DFA, which accepts a language $L \sse \oneset{b,\ov b}^*$,
	whether or not 
	$L$ is empty. 
	
	Now let $L_1 = a^*L\ov a^k$. The hairpin completion
	\begin{equation*}
		\Hk(L_1,\es) = \set{a^iw\ov a^j}{i\geq j\geq k \land w\in L}
	\end{equation*}
	is regular if and only if $L$ is empty (because $L \sse \oneset{b,\ov b}^*$).	
\end{proof}

The difficult part is to show that deciding regularity of 
\ccH is in  \NL. This is subject of the rest of
this section.

\subsection{\Svlstd{}s}\label{svlstd}
A \svlstd is performed by a \nd $\log$-space Turing machine
which may stop  on every input $w$ with some  output $r(w)$. \emph{Single-valued} means that, in case that the machine stops on input $w$, 
 the output is always the same, independently of \nd moves during the 
    computation. Thus, $w \mapsto r(w)$ is a well-defined function from words to words. A \svlstd is a \emph{reduction} {}from  a language $L$ to $L'$, 
if we have $w \in L \iff  r(w) \in L'$.

The following lemma belongs to folklore. Its proof is exactly the 
same as for the standard case of deterministic $\log$-space reductions \cite{HU}
and therefore omitted.

\begin{lemma}\label{lem:svlstd}
Let $L' \in $ \NL and assume that there exists a \svlstd from $L$ to $L'$. Then we have $L \in $ \NL, too. 
\end{lemma}

Due to \reflem{lem:svlstd} we are free to use several \svlstd{}s in order to 
enrich the input. 

\subsection{Bridges}\label{bridges}
Let $\cQ_1, \cQ_2$ be the state sets as  fixed 
by \refthm{thm:main}. For every quadruple $(p_1,p_2,q_1,q_2) \in \cQ_1 \times \cQ_2\times \cQ_1\times \cQ_2$
we define a regular language $B(p_1,p_2,q_1,q_2)$ as follows: 
$$B(p_1,p_2,q_1,q_2) = \set{\bet\in \Sig^*}{p_1\cdot \bet = q_1 \wedge p_2\cdot \ov \bet = q_2}.$$

We say that a quadruple $(p_1,p_2,q_1,q_2)$ is a \emph{bridge}, if  $B(p_1,p_2,q_1,q_2)\neq \es$.
The idea behind  this notation is that $B(p_1,p_2,q_1,q_2)$ closes 
a gap between pairs $(p_1,p_2)$ and $(q_1,q_2)$. 
For a bridge $(p_1,p_2,q_1,q_2)$ the words $\bet \in B(p_1,p_2,q_1,q_2)$
correspond later exactly to the $\bet$-part in \reffig{whatanicehairpin}. 

\begin{lemma}\label{lem:bridge}
	There is a \svlstd  which
	outputs the table of all bridges.
\end{lemma}

\begin{proof}
	Graph reachability and its complement are solvable in \NL. Therefore
	we can decide for each quadruple
	$(p_1,p_2,q_1,q_2) \in \cQ_1 \times \cQ_2\times \cQ_1\times \cQ_2$
	if it is a bridge, and we can output $(p_1,p_2,q_1,q_2)$ in the affirmative case.
\end{proof}

\subsection{The NFA $\cA$}\label{secnfa}
Next, we  construct an NFA, which is called simply $\cA$, and 
we explore properties of this NFA.
The NFA $\cA$ uses $k+1$ levels (or layers) of a product automaton over $\cQ \times \cQ_1\times \cQ_2 \sse 
\cQ_1 \times \cQ_2\times \cQ_1\times \cQ_2$ where $ \cQ$ has been defined 
as in Section~\ref{secpn}.
Hence, the number  of states is at most $(k+1)n^4$ which is in $\Oh(n^4)$.  

Formally, we use a \emph{level} for each $\ell$ with $0 \leq \ell \leq k$, hence there are $k+1$ levels.
By $[k]$ we denote in this paper the set $\oneset{0 \lds k}$. 
Define 
$$\cQ_\cA= \set{((p_1,p_2),q_1,q_2, \ell) \in \cQ \times \cQ_{1} \times \cQ_{2} \times [k]} 
{(p_1,p_2,q_1,q_2) \text{ is a bridge}} $$ as the state space of an NFA called $\cA$.

We call  a state $((p_1, p_2) ,q_1,q_2,\ell)$ a \emph{bridge at level $\ell$}, and we  keep in mind that there exists a word $w$
such that $p_1\cdot w= q_1$ and $p_2\cdot \ov w= q_2$. 
Frequently (and by a slight abuse of language) we call  a state $((p_1, p_2) ,q_1,q_2,\ell)$ simply a \emph{bridge}, too. 
Bridges at level $\ell$ are also denoted by $(P,q_1,q_2,\ell)$ with $P= (p_1, p_2)\in \cQ$, 
$q_i \in Q_i$, $i=1,2$, and  $\ell \in [k]$.
Bridges at different levels play a  central r{\^{o}}le in the following. 

Let $a\in \Sigma$.  The $a$-transitions in the NFA 
are given by the following arcs: 
\begin{alignat*}{2}
(P,\;q_1\cdot \ov a ,\; q_2 \cdot \ov a ,0) &\ras{a}{} (P\cdot a ,\;q_1,q_2,0)& \text{ for }
q_i\cdot \ov a \notin \cF_i, \, i = 1,2,\\
(P,\;q_1\cdot \ov a ,\;q_2 \cdot  \ov a ,0) &\ras{a}{} (P\cdot a ,\;q_1,q_2,1)& \text{ for }
q_1\cdot \ov a \in \cF_1 \text{ or }q_2\cdot \ov a \in \cF_2, \\
(P,\;q_1\cdot \ov  a ,\;q_2 \cdot  \ov a ,\ell) &\ras{a}{} (P\cdot a ,\;q_1,q_2,\ell + 1)& \text{ for }
1 \leq \ell < k.
\end{alignat*}

Thus, for the $P$-component an $a$-transition behaves as in a usual product automaton, but for the $q_1$- and $q_2$-components we move backwards using the $\ov a$-transitions in the original automata. This is why the resulting automaton  $\cA$
is \nd.

Observe that  no state
of the form $(P,q_1,q_2,0)$ with $q_1 \in \cF_1$ or $q_2 \in \cF_2$ has an outgoing 
arc to level zero; we must switch to level one. There are no outgoing arcs on level $k$, and 
for each tuple $(a, P,q_1,q_2,\ell) \in \Sigma \times \cQ \times \cQ_{1} \times \cQ_{2}\times [k-1] $ there exists 
at most  one arc $(P,q_1',q_2',\ell) \ras{a}{} (P\cdot a ,q_1,q_2,\ell')$.
Indeed, the $P\cdot a$ is  determined by $P$ and the letter $a$, and the  triple $(q_1',q_2',\ell')$ is determined by $(q_1,q_2,\ell)$ and the letter $a$. Not all such arcs exist in $\cA$, because
$(P,q_1',q_2',\ell)$ might be a bridge whereas  $(P\cdot a ,q_1,q_2,\ell')$ is not.
(Observe however that if   $(P\cdot a ,q_1,q_2,\ell')$ is a bridge, then 
$(P,q_1',q_2',\ell)$ is a bridge, too.) 

The set of initial states $\cI$ contains all bridges at level zero of the form
$(Q_0,q'_1,q'_2,0)$ with $Q_0 = (q_{01},\, q_{02})$. 
The set of final states $\cF$ is given by all bridges $(P,q_1,q_2,k)$ at level $k$.

This concludes the definition of the NFA $\cA$. For an example and a graphical presentation of the NFA, see Figure~\ref{steffen}.

\begin{remark}\label{rem:prec}
By Lemma~\ref{lem:bridge}, the NFA $\cA$ can be computed by a \svlstd. Thus, we have direct access
to $\cA$ and henceforth we assume that $\cA$ is also written on  the input tape. 
\end{remark}

\begin{figure}[htb]
  \centering
  \begin{tikzpicture}[shorten >=1pt,node distance=2cm,auto,initial text=,%
  initial distance=4mm,bend angle=45,scale=.85]
    \tikzstyle{every node}=[scale=.85]
    \tikzstyle{every loop}=[distance=.5cm]
	\node at (0,5.5)	[state,initial]		(A0)						{$q_{01}$};
	\node 			[state]				(A1)		[right of=A0]	{$p_1$};
	\node 			[state, accepting]	(A2)		[right of=A1]	{$f_1$};
	\node 			[state]				(A3)		[below of=A1]	{$\dead_1$};
	\node			[above of=A1,node distance=1.5cm]			{$L_1 = a^*(b+ \ov b)\ov a$};
    
    \path [->]	(A0)	edge	[loop above]	node			{$a$}			()
    					edge					node			{$b,\ov b$}		(A1)
    					edge	[bend right]	node	[swap]	{$\ov a$}		(A3)
    			(A1)	edge					node			{$\ov a$}		(A2)
	    				edge					node			{$a,b,\ov b$}	(A3)
    			(A2)	edge	[bend left]		node			{$\Sigma$}		(A3)
    			(A3)	edge	[loop below]	node			{$\Sigma$}		();

    \node at (7,5.5)	[state,initial]		(B0)					{$q_{02}$};
    \node			[state]				(B1)	[right of=B0]	{$p_2$};
    \node			[state, accepting]	(B2)	[right of=B1]	{$f_2$};
    \node			[state]				(B3)	[below of=B1]	{$\dead_2$};
	\node			[above of=B1,node distance=1.5cm]			{$\ov{L_2} = a^*\ov b\ov a$};
    
    \path [->]	(B0)	edge	[loop above]	node			{$a$}			()
    					edge					node			{$\ov b$}		(B1)
    					edge	[bend right]	node	[swap]	{$\ov a,b$} 	(B3)
    			(B1)	edge					node			{$\ov a$}		(B2)
    					edge					node			{$a,b,\ov b$}	(B3)
    			(B2)	edge	[bend left]		node			{$\Sigma$}		(B3)
    			(B3)	edge	[loop below]	node			{$\Sigma$}		();

%
    \tikzstyle{every state}=[rectangle]
    \tikzstyle{every pin}=[pin distance=4mm]
    \tikzstyle{every pin edge}=[shorten <=1pt]
    \tikzstyle{init}=[pin={[pin edge={<-}]170:}]
    
    \node at (0,0)		[state,initial]		(A)		{$(Q_0,\dead_1,\dead_2,0)$};
    \node at (3,0)		[state,init]		(B)		{$(Q_0,f_1,f_2,0)$};
    \node at (6,0)		[state,accepting]	(B1)	{$(Q_0,p_1,p_2,1)$};
    \node at (7.25,0)	[right]						{$B(q_{01},q_{02},p_1,p_2)= b$};
    \node at (3,1)		[state,initial]		(C)		{$(Q_0,f_1,\dead_2,0)$};
    \node at (6,1)		[state,accepting]	(C1)	{$(Q_0,p_1,\dead_2,1)$};
    \node at (7.25,1)	[right]						{$B(q_{01},q_{02},p_1,\dead_2)= aa^+b + a^*\ov b$};
    \node at (6,2)		[state,accepting]	(C2)	{$(Q_0,p_1,f_2,1)$};
    \node at (7.25,2)	[right]						{$B(q_{01},q_{02},p_1,f_2)= ab$};
    \node at (3,-1)	[state,initial]		(D)		{$(Q_0,\dead_1,f_2,0)$};
    \node at (6,-1)	[state,accepting]	(D1)	{$(Q_0,\dead_1,p_2,1)$};
    \node at (7.25,-1)	[right]						{$B(q_{01},q_{02},\dead_1,p_2)= b\ov a\ov a^+$};
    \node at (6,-2)	[state,accepting]	(D2)	{$(Q_0,f_1,p_2,1)$};
    \node at (7.25,-2)	[right]						{$B(q_{01},q_{02},f_1,p_2)= b\ov a$};
    \node at (-.8,2) 	{$\cA$:};

    \path [->]	(A)	edge										node			{$a$}	(B)
					edge	[loop,out=125,in=55,distance=1cm]	node		 	{$a$}	(A)
					edge										node			{$a$}	(C)
					edge										node	[swap]	{$a$}	(D)
				(B)	edge										node			{$a$}	(B1)
				(C)	edge										node			{$a$}	(C1)
					edge										node			{$a$}	(C2)
				(D)	edge										node			{$a$}	(D1)
					edge										node	[swap]	{$a$}	(D2);
  \end{tikzpicture}
  \caption{DFAs for $L_1$ and $\ov{L_2}$ and the resulting NFA $\cA$ with 4 initial states
  and 5 final states associated to the (\lcf) \hpc
  $\ccH = a^+ b \ov a ^+ \cup \{a^s \ov b \ov a^t\mid s \geq t \geq 1\}$ with $k=1$. }\label{steffen}
\end{figure}

The next result shows the unambiguity of paths in the automaton $\cA$.
It is a crucial property. 

\begin{lemma}\label{unam}
Let $w \in \Sig^*$ be the label of a path in $\cA$ from a bridge 
$A= (P,p_1,p_2,\ell) $ to $A' = (P',p_1',p_2',\ell')$,
then the path is unique. This means that $B=B'$ whenever $w = uv$ and 
\begin{align*}
A &\ras{u}{} B 
  \ras{v}{} A', &&
A \ras{u}{} B' 
  \ras{v}{} A'.
\end{align*}
\end{lemma}

\begin{proof}
 It is enough to consider $u = a \in \Sig$. 
 Let $B = (Q,q_1,q_2,m)$. Then we have 
 $Q= P \cdot a$ and $q_i= p_i' \cdot \ov v$.
 If $\ell= 0$ and  $p_i\notin \cF_i$ for $i = 1,2$, then 
 $m= 0$, too; otherwise $m= \ell +1$. Thus, $B$ is determined by 
 $A$, $A'$, and $u$, $v$. We conclude $B=B'$.
\end{proof}

We will now show that the automaton $\cA$ encodes the
hairpin completion in a natural way.
For languages $U$ and $V$ we define the language $V^U$ as follows:
$$V^U  = \set{uv\ov u}{u \in U, \, v \in V}.$$
Clearly, if $U$ and $V$ are regular, then $V^U$ is \lcf, but not regular, in general.
(The notation $V^U$ is adopted {}from group theory where exponentiation denotes conjugation 
and the canonical involution refers to taking inverses.) 

\begin{lemma}\label{lem:str}
For each pair $\tau= (I,F)\in \cI \times \cF$ with $F = ((d_1,d_2),e_1,e_2,k)$
let $R_\tau $ be the (regular) set of words which label a path from the initial bridge 
$I$ to the final bridge $F$, and let $B_\tau = B(d_1,d_2,e_1,e_2)$.

The  \hpc $\ccH$ is a disjoint union 
$$\ccH= \bigcup_{\tau \in \cI \times \cF}B_\tau^{R_\tau}.$$
Moreover, for each word $w\in B_\tau^{R_\tau}$ there is a 
unique factorization $w = \rho\bet\ov\rho$ with $\rho \in R_\tau$ and 
$\bet \in B_\tau$. 
\end{lemma}

\begin{proof}
Let $w\in\ccH$. There exists  some factorization $w=\gabag$ such that $\abs\alp=k$
and there are runs as in \reffig{firstrun} in the original DFAs $\cA_1$ and $\cA_2$
where $f_1'\in\cF_1$ or $f_2'\in\cF_2$ (or both): 
\begin{figure}
\begin{align*}
  L_1: \quad& q_{01} \ras{\gam} c_1' \ras{\alp} d_1' \ras {\bet }{}
    e_1' \ras {\ov \alp}{} f_1' \ras {\ov \gam }{} q_1', \\
  \ov{L_2}: \quad& q_{02} \ras{\gam} c_2' \ras{\alp} d_2'\ras {\ov \bet }{}
    e_2' \ras {\ov \alp}{} f_2' \ras {\ov \gam}{} q_2'
\end{align*}
\caption{Some run defined by $w\in \ccH$}
\label{firstrun}
\end{figure}

Choosing among all these runs the length $\abs{\ov \gam}$ to be minimal, we see that 
we actually find the following picture according to \reffig{urun}.
In other words, either $\gaba$ is the longest prefix of $w$ belonging to $L_1$ or
$\abag$ is the longest suffix of $w$
belonging to $L_2$, or both. 
The difference to the precedent figure is is that between $f_i$ and $q'_i$ ($i= 1,2$) we never enter a final state. 
\begin{figure}
 \begin{align*}
  L_1: \quad& q_{01} \ras{\gam} c_1 \ras{\alp} d_1 \ras {\bet }{}
    e_1 \ras {\ov \alp}{} f_1 \RAS {\ov \gam }{} q_1', \\
  \ov{L_2}: \quad& q_{02} \ras{\gam} c_2 \ras{\alp} d_2\ras {\ov \bet }{}
    e_2 \ras {\ov \alp}{} f_2 \RAS {\ov \gam}{} q_2'
\end{align*}
\caption{The unique run defined by $w\in \ccH$ with $\abs{\ov \gam}$ minimal}
\label{urun}
\end{figure}
By the definition 
of the NFA $\cA$ we see that $\rho = \gam  \alp$ is the unique prefix of $w$ such that
$w = \rho\bet\ov\rho$ with $\rho \in R_\tau$ and 
$\bet \in B_\tau$ for some $\tau$. 
Now, as the length $\abs{\ov \gam}$ is fixed by $w$, we see that all states 
$c_i$, $d_i$, $e_i$, $f_i$, and $q_i'$ are uniquely defined by $w$ for $i= 1,2$. 
Thus, there is a unique $\tau \in \cI \times \cF$ with 
$w \in B_\tau^{R_\tau}.$ More precisely, we have: 
$$\tau = (((q_{01}, q_{02}), q_{1}', q_{2}',0), \, ((d_1, d_2), e_1, e_2,k)).$$ 
\end{proof}

\subsection{First Tests}\label{test}
By construction, the automaton $\cA$ accepts the union of the languages $R_\tau$ as defined
in Lemma~\ref{lem:str}. If the accepted language is finite then all
$R_\tau$ are finite and hence all $B_\tau^{R_\tau}$ are regular. 
This leads immediately to the following result: 

\begin{proposition}\label{prop:onesided}
It can be decided in \NL whether or not the accepted language of the  NFA $\cA$ is finite.
If 
the accepted language is finite, then the \hpc $\ccH$ is regular.
\end{proposition}

\begin{proof}
 To see that the accepted language is infinite it is enough to guess a path 
 from an initial state to final one which uses some (guessed)  state at least twice.
 Since \NL is closed under complementation the finiteness test is  possible in \NL, too. The second assertion follows {}from Lemma~\ref{lem:str}.
\end{proof}

We check this property (although strictly speaking Test~0  is redundant): 

\begin{test}{0}
	Decide in \NL whether or not $L(\cA)$ is finite. 
	If it is finite, then stop with the output that $\ccH$ is regular. 
\end{test}

For convenience we may assume in the following that $\cA$ accepts an infinite
language and that all states are reachable from an initial bridge
and lead to some final bridge.

For sake of completeness let us state another result which shows that
deciding regularity of the one-sided \hpc is somewhat easier, because 
 the finiteness condition is also necessary in this case.  
However, as we neither use this result nor does it change the \NL-completeness of the problem, we leave the
proof of \refprop{prop:trueonesided} to the interested reader. 
\begin{proposition}\label{prop:trueonesided}
If $L_1$ or $L_2$ is finite, but the accepted language of $\cA$  is infinite,
then the \hpc $\ccH$ is not regular. 
\end{proposition}

Let $K$ be the set of non-trivial strongly connected components of the automaton $\cA$ 
(read as a directed graph). Every non-trivial strongly connected component is 
on level 0 and, moreover, as $\cA$ accepts an infinite
language, there is at least one. For $\kappa \in K$ let $N_\kappa$
be the number of states in the component $\kappa$. We have $N_\kappa = \abs \kappa \leq n^4$.

The next lemma tells us that for a regular \hpc $\ccH$ every
strongly connected component $\kappa\in K$ is a simple cycle.

\begin{lemma}\label{lem:loop}
	Let the hairpin completion \ccH be regular,
	$A \ras {v_A} A$  be a path in a strongly connected component $\kappa$ with $1\leq \abs {v_A} \leq N_\kappa$,
	and let $A\ras w F$ be a path in $\cA$ {}from $A$ to a final bridge $F$.
	Then the word $w$ is a prefix of some word in $v_A^+$.

	In addition, the  word $v_A$ is uniquely defined by the conditions  $A\ras {v_A} A$ and $1\leq \abs v_A\leq N_\kappa$.
	The loop $A\ras {v_A} A$ visits every other state $B\in\kappa$ exactly once.
	Thus it builds a Hamiltonian cycle of $\kappa$ and $\abs{v_A} = N_\kappa$.
\end{lemma}

\begin{proof}
Let  $A\ras{v}{}A$ be some  non-trivial loop. We see that $A$ is on level zero. 
Consider a path labeled by $w$ from $A$ to a final bridge $F=((p_1,p_2),q_1,q_2,k)$.
By assumption,  all states in $\cA$ are reachable from some initial state. Thus,  we find 
a word $u$ such that the automaton $\cA$  accepts $uv^i w$ for all $i \geq 0$.
We see next that 
$uv^i w \bet \ov w \ov v^i \ov u \in \ccH$ for all $i \geq 0$
and all $\bet \in  B(p_1,p_2,q_1,q_2)$.
As $\ccH$ is regular, there are $s,t\in \N $ with $uv^s w \bet \ov w \ov v^{s+t} \ov u \in \ccH$
and $t > \abs{w \bet}$, by pumping.
This means that the 
\hpc is forced to use  a suffix in $L_2$, because the 
longest prefix belonging to $L_1$ is too short to create the \hpc.  Due to the definition of $\cA$ we conclude 
that $uv^s w$ must be a prefix of  $uv^{s+t} w$. This implies that $w$ is a prefix of  $v^{t}$
and thus the first statement of our lemma.

Let $v_A$ be some shortest word such that $A\ras {v_A} A$.
Observe first that $\abs{v_A} \leq N_\kappa$.
Now, let $A\neq B \in \kappa$ and $A\ras{v'}{} B\ras{v''}{}A$. 
For some $i, j >0$ we have $\abs {v_A^i} = \abs {(v'v'')^j}$. Thus, $v_A^i = (v'v'')^j$
by the first statement.
By the unique-path-property stated in Lemma~\ref{unam} we obtain that
the loop $A\ras{(v'v'')^j}A$ just uses the shortest loop $A\ras{v_A}A$ several times.
In particular, $B$ is on the shortest loop around $A$. This yields $\abs {v_A} \geq N_\kappa$
and hence the second statement.
\end{proof}

\begin{example}
	In the example given in Figure~\ref{steffen} the state $(Q_0,\dead_1,\dead_2,0)$
	forms the only strongly connected component and the corresponding path is labeled with $a$.
	As one can easily observe the automaton $\cA$ satisfies the properties stated in Lemma~\ref{lem:loop}
	(even though the \hpc is not regular).
\end{example}

Due to the technique of \svlstd{}s we may assume that the set of non-trivial strongly connected components $K$ is part of the input. Moreover, 
for each state $A$ and $\kappa \in K$ we know whether 
or not $A\in \kappa$, and we know $N_\kappa = \abs \kappa$.

The next test tries to falsify the property of \reflem{lem:loop}. 
Hence it gives a sufficient condition that $\ccH$ is not regular. 

\begin{test}{1}
	Guess some  state $A$ and $\kappa \in K$ with $A \in \kappa$, a letter $a\in \Sig$,
	and a position $1\leq m\leq N_\kappa$ such that:
	\begin{enumerate}[1.)]
		\item There is a path $A\ras v A$ where $m \leq \abs v\leq N_\kappa$ and $v(m) = a$.
		\item There is a path $A\ras w F$ where $w(i\cdot\abs{v}+m) \neq a$ for some $i\in \N$ with $ 1 \leq i\cdot\abs{v}+m \leq \abs w$. 
	\end{enumerate}
	If such a triple $(A,a,m)$ exists, then  output that \ccH is not regular.
\end{test}

The correctness of Test~1 follows by Lemma~\ref{lem:loop} and,
because for the existence of paths 1.)\ and 2.) we only have to remember the triple $(A,a,m)$,
Test~1 can be performed in \NL.

\begin{remark}\label{rem:test1}
	We can perform Test~1 in \NL and in case  it yields that
	the hairpin completion \ccH is not regular, we can stop.
	Henceforth, we assume that the algorithm did not stop during
	Test~1 and that every strongly connected component $\kappa\in K$ is a simple cycle.
	Performing another \svlstd we may assume that for each $A \in \kappa$
	the word $v_A$ is attached to $A$ and each $v_A$ is part of the input. 
\end{remark}

\subsection{Second and Third Test}\label{test2}
We fix a bridge $A = ((p_1,p_2),q_1,q_2)$ in a strongly connected component.
We let $v = v_A$ as defined in Lemma~\ref{lem:loop} and
let $\alp$ be the prefix of length $k$ of some long enough word in $v^+$.
(By \refrem{rem:test1} the word $v$ is written in plain form on the input tape.) 
By $u$ we denote some word leading from an initial bridge to $A$.
(The \NL algorithm does not know $u$, but it knows that it exists.) 
The main idea is to investigate runs through the DFAs for $L_1$ and
$\ov{L_2}$ where $s,t\geq n$ according to \reffig{otto}.
Recall that $n$ refers to the original input size, thus $n \geq \abs{\cQ_i}$ for $ i = 1,2$. 

\begin{figure}[h]
\begin{alignat*}{3}
  &L_1: &\quad&q_{01} \ras{u}{} p_1 \ras{v^s}{} p_1 \ras{x}{} c_1 \ras{y}{} d_1
    \ras{\ov \alp}{} e_1 \RAS{\ov v^{n}}{} &&q_1 \RAS{\ov v^*}{} q_1 \RAS{\ov u}{} q_1' \\
  &\ov{L_2}: &&q_{02} \ras{u}{} p_2 \ras{v^t}{} p_2 \ras{\alp}{} c_2 \ras{\ov y}{} d_2
    \ras{\ov x} e_2 \RAS{\ov v^{n}}{} &&q_2 \RAS{\ov v^*}{} q_2 \RAS{\ov u}{} q_2'
\end{alignat*}
\caption{Runs through $\cA_1$ and $\cA_2$ based on the loop
$A\xrightarrow{\mspace{5mu}v\mspace{5mu}} A$}
\label{otto}
\end{figure}

We investigate the case where $uv^s xy\ov\alp \ov v^t\ov u\in\ccH$ for all $s\geq t$ and
where (by symmetry) this property is due to the longest prefix belonging to $L_1$
(hence $e_1\in\cF_1$).

The following lemma is rather technical. The notations are however
chosen to fit exactly to \reffig{otto}. 

\begin{lemma}\label{lem:egil}
Let $x, y \in \Sig^*$ be words and $(d_1,d_2) \in \cQ_1 \times \cQ_2$ with
the following properties: 
\begin{enumerate}[1.)]
	\item $\alp \leq x$ and $x < v\alp$.
	\item $y\in B(c_1,c_2, d_1,d_2)$, where
		$c_1 = p_1 \cdot x$ and $c_2 = p_2 \cdot \alp$, and
		$x$ is the longest common prefix of $xy$ and $v\alp$.
	\item $e_1 = d_1\cdot \ov\alp \in \cF_1$ is a final state,
		$q_1 = e_1 \cdot  \ov {v}^n$, and during
		the computation of $e_1 \cdot  \ov {v}^n$ we do not enter a final state in $\cF_1$.
	\item $e_2 = d_2 \cdot \ov{x}$ and $q_2 = e_2 \cdot \ov {v}^n$. Moreover, 
		during the computation of $e_2 \cdot \ov{v}^n$ we do not enter a final state in $\cF_2$
		(but $e_2\in\cF_2$ is possible).
\end{enumerate} 
If $\ccH$ is regular, then there exists a factorization $xy\ov \alp \ov v = \mu \del\bet \ov \del \ov \mu$ where
$\abs\del = k$ and $p_2\cdot \mu\del\ov\bet\ov\del\in\cF_2$
(which implies $\del\bet\ov\del\ov\mu\ov v^*\ov u \sse L_2$). 
\end{lemma}

\begin{proof}
The conditions imply that 
$ uv^s xy  \ov \alp \ov v^{t} \ov u \in \ccH$ for all $s \geq t \geq n$.
Moreover, by 3.)\ the \hpc can be achieved with a prefix in $L_1$
and the longest prefix of $uv^s xy  \ov\alp \ov v^{t} \ov u$ belonging to 
$L_1$ is $uv^s xy \ov \alp$. 

If $\ccH$ is regular, then
we have $ uv^{s} xy  \ov \alp \ov v^{s+1} \ov u \in \ccH$, too,
as soon as $s$ is large enough, by a simple pumping argument.
For this \hpc we must use a suffix belonging to $L_2$. 
For  $y= \e$ this follows {}from $x < v\alp$. 
For  $y\neq \e$ we use $x < v\alp$
and additionally  that the word $xa$
with $a=y(1)$ is not a prefix of $v\alp$.

By 4.)\ the longest suffix of $ uv^{s} xy \ov\alp \ov v^{s+1} \ov u$ belonging to 
$L_2$ is a suffix of $xy  \ov\alp \ov v^{s+1} \ov u$. 
Thus, we can write 
$$ uv^{s} xy  \ov \alp \ov v^{s+1} \ov u = uv^{s} xy  \ov\alp \ov v \ov v^{s} \ov u
=  uv^{s} \mu \del\bet \ov \del \ov \mu \ov v^{s} \ov u$$
where  $\del\bet \ov \del \ov \mu \ov {v}^{s} \ov u \in L_2$
and $\abs\del = k$. We obtain $xy\ov \alp \ov v = \mu \del\bet \ov \del \ov \mu$. 

(Recall that our second DFA $\cA_2$ accepts $\ov {L_2}$.)
Hence, as $p_2 = q_{02} \cdot u$ and $p_2 = p_2 \cdot v$,  we see that 
$p_2\cdot \mu\del\ov\bet\ov\del \in\cF_2$. 

We conclude as desired:
if $\ccH$ is regular, then $p_2\cdot \mu\del\ov\bet\ov\del \in\cF_2$.
\end{proof}

\begin{example}
Let us take a look at Figure~\ref{steffen} again.
Let $A = (Q_0,\dead_1,\dead_2,0)$, $v = a$ and $u =\e$.
If we choose $x = a$, $y=\ov b$ and $(d_1,d_2) = (p_1,p_2)$ we can
see, that conditions 1.)\ to 4.)\ of Lemma~\ref{lem:egil} are satisfied but there
is no factorization $a\ov b\ov a\ov a = \mu\del\bet\ov\del\ov\mu$ with $\abs\del = k$ such that
$\del\bet\ov\del\ov\mu\ov u\in L_2$. Hence, the \hpc is not regular.
\end{example}

The next lemma yields another sufficient condition that $\ccH$ is not regular. 

\begin{lemma}\label{lem:final}
The existence of words $x, y \in \Sig^*$ and states
$(d_1,d_2) \in \cQ_1 \times \cQ_2$ satisfying 1.) to 4.) of Lemma~\ref{lem:egil},
but where for all factorizations $xy\ov \alp \ov v = \mu \del\bet \ov \del \ov \mu$
we have $p_2\cdot \mu\del\ov\bet\ov\del\notin\cF_2$
can be decided in \NL. 
\end{lemma}

\begin{proof}
It is enough to perform either Test~2 or~3 below (non-deterministically chosen) and
to prove the \NL performance of these tests.
The tests distinguish whether the word $y$ is empty or non-empty. 

\begin{test}{2}
Decide the existence of a word $x \in \Sig^*$  and states
$(d_1,d_2) \in \cQ_1 \times \cQ_2$ satisfying 1.)\ to 4.)\ of Lemma~\ref{lem:egil}
with $y =\e$, but where for all factorizations $x\ov \alp \ov v = \mu \del\bet \ov \del \ov \mu$ we have  
$p_2\cdot \mu\del\ov\bet\ov\del\notin\cF_2$. 
If we find such a situation, then output that $\ccH$ is not regular.
\end{test}

\begin{test}{3}
Decide the existence of words $x, y \in \Sig^*$ 
with $y \neq \e$ and states
$(d_1,d_2) \in \cQ_1 \times \cQ_2$ satisfying 1.)\ to 4.)\ of Lemma~\ref{lem:egil},
but where for all factorizations $xy\ov \alp \ov v = \mu \del\bet \ov \del \ov \mu$ we have
$p_2\cdot \mu\del\ov\bet\ov\del\notin\cF_2$. 
If we find such a situation, then output that $\ccH$ is not regular. 
\end{test}
  
The correctness of both tests follows by Lemma~\ref{lem:egil} and
they can be performed as follows:
For both tests we guess the length of a word $x$ which satisfies 1.) and which is therefore a prefix of $v\alp$. Thus we can remember $x$, because $v\alp$ is available by the input. 
We guess states $(d_1,d_2)\in \cQ_1\times\cQ_2$,
and verify that conditions 3.)\ and 4.)\ hold, which is easy because we can reconstruct $x$. 
For Test~2 we check that $p_1\cdot x = d_1$ and $p_2\cdot \alp = d_2$.
Then we have to test whether for all factorizations $x\ov\alp\ov v = \mu\del\bet\ov\del\ov\mu$
with $\abs\del = k$
the condition $p_2\cdot \mu\del\ov\bet\ov\del\notin\cF_2$ holds.
This can easily be done in \NL because we have full access to the word $x\ov\alp\ov v$. 

Test~3 is a bit more tricky.
We guess $a\in\Sig$ and we check  that $xa$ is not a prefix of $v\alp$.
We have to verify that a path from $c_1$ to $d_1$ exists which is labelled by some non-empty word $y\in a\Sig^*$
and that a path from $c_2$ to $d_2$ exists which is labelled by $\ov y$.
This can be achieved by a graph reachability algorithm which uses forward edges in the
DFA of $L_1$ and simultaneously uses backwards edges in the DFA of $\ov{L_2}$.
Now, in a factorization $xy\ov\alp \ov v =\mu\del\bet\ov\del\ov\mu$ we cannot have that
$x$ is a proper prefix of $\mu\del$ otherwise $xa$ would be a prefix of $v\alp$.
But this was excluded by the choice of $a$. Thus, $\mu\del$ is a prefix of $x$
and $\ov \del \ov \mu $ is a suffix of $\ov x$. 
This means, to ensure that there is no
factorization with $p_2\cdot \mu\del\ov\bet\ov\del\in\cF_2$,
we do not need to remember the word $y$.  We just compute $d_2\cdot \ov x$ and during this computation we
validate that there are no final states in $\cF_2$ after $k$ or more steps.
\end{proof}

We claim that, if all three tests did not yield that the
hairpin completion \ccH is not regular,
then the hairpin completion is indeed regular.
This will complete the proof of Theorem~\ref{thm:main}.

\begin{lemma}\label{lem:finaltests}
Suppose no outcome of Tests 1, 2, and 3 is ``not regular''. Then the \hpc $\ccH$ is regular. 
\end{lemma}

\begin{proof}
Let  $\pi  \in \ccH$. Write 
$\pi =  \gam \alp  \bet \ov \alp \ov \gam$ with $\abs \gam $ minimal such 
that either $\gam \alp \bet \ov \alp \in L_1$ or 
$\alp   \bet \ov \alp \ov \gam \in L_2$.
By symmetry we assume $\gam \alp \bet \ov \alp \in L_1$.
We may also assume that $\abs \gam> 2n^4$ 
(cf.\ Proposition~\ref{prop:onesided}
and Test~0).
We can factorize $\gam = uvw$ with $\abs {uv} \leq n^4 $ and $1 \leq \abs v \leq \abs w $
such that there are runs as in Figure~\ref{karl}.

\begin{figure}[h]
\begin{alignat*}{3}
	&L_1: &\quad&q_{01} \ras{u}{} p_1 \ras {v}{} p_1 \ras{w\alp\beta\ov\alp}{}&&
		f_1 \RAS {\ov w}{} q_1 \RAS {\ov v}{} q_1 \RAS {\ov u }{} q_1' \\
  &\ov{L_2}: &&q_{02} \ras{u}{} p_2 \ras {v}{} p_2 \ras{w\alp\ov\beta\ov\alp}{}&&
  f_2 \RAS {\ov w}{} q_2 \RAS {\ov v}{} q_2 \RAS {\ov u }{} q_2'
\end{alignat*}
\caption{Runs through $\cA_1$ and $\cA_2$ for the word $\pi$.
	We assume $f_1\in\cF_1$.}
\label{karl}
\end{figure}

We infer from Test~1 that $w \alp$ is a prefix of some word in $v^+$.
We may assume that $w\in v^+$ by adjusting the 
choices of $u$, $v$, and $w$.
(Possibly, $u$ gets longer but it is still shorter than $n^4$,
$v$ is transposed, and $w$ gets shorter.)

Hence, we can write $w \alp \bet = v^m xy$ with $m\geq 0$
such that  $v^m x$ is the maximal common prefix of $w\alp \bet $ and some word in $v^+$ with $\alp \leq x  < v\alp$. 

We see that for some $s \geq t \geq 0$ we can write 
$$ \pi = uv^s xy  \ov \alp \ov v^{t} \ov u.$$

Moreover, $ uv^s xy \ov\alp \ov v^{t} \ov u \in \ccH$ for all $s \geq t \geq 0$.
There are only finitely many choices for $u,v,x$ (due to the lengths bounds) 
and for each of them there is a regular set 
$R_y$ associated to the finite collection 
of bridges such that 
$$\pi \in \set{uv^s xR_y \ov\alp \ov v^{t} \ov u}{s \geq t \geq 0} \sse \ccH.$$

More precisely, we can choose $R_y = \oneset{\e}$ for $y=\e$, and otherwise we can choose 
$$R_y \in \set{B(c_1,c_2,d_1,d_2) \cap a \Sig^*}{
(c_1,c_2,d_1,d_2) \text{ is a bridge and } a \in \Sig}.$$

Note that the sets $\set{uv^s xR_y  \ov\alp \ov v^{t} \ov u}{s \geq t \geq 0}$ 
are not regular, in general. 
If we bound however the exponent $t$ by $n$, then 
the finite union 
$$\bigcup_{0 \leq t\leq n}\set{uv^s xR_y  \ov\alp \ov v^{t} \ov u}{s \geq t}$$
becomes regular. Thus, we may assume that $t > n$.
Let $e_2 = p_2 \cdot \alp\ov y \ov x$. We have $e_2 \cdot \ov v^{n} = q_2$
and, if there is a final state during the computation of $e_2\cdot \ov v^{n}$, then
for all $t \geq s\geq n$ and $y\in R_y$ we have that $uv^s xy \ov\alp \ov v^{t} \ov u \in \ccH$,
due to a suffix in $L_2$, and
$uv^n v^+ xR_y  \ov\alp \ov v^{+} \ov v^{n} \ov u \sse \ccH$.

Otherwise Test~2 or~3 tells us that for all $y\in R_y$ the word
$xy\ov\alp \ov v$ has a factorization $\mu\delta\nu\ov\delta\ov\mu$
such that $\abs\delta = k$ and $p_2\cdot \mu\del\ov\nu\ov\del\in\cF_2$.
The paths $q_{02} \cdot u = p_2$ and $p_2\cdot v = p_2$
yield $\delta\nu\ov\delta\ov\mu \ov v^*\ov u \sse L_2$
and, again,
$uv^n v^+ x R_y \ov\alp \ov v^{+} \ov v^{n} \ov u \sse \ccH$.

The \hpc $\ccH$ is a finite union of regular languages
and hence it is regular itself.
\end{proof}

\section{Conclusion and open problems}\label{secopen}
We have shown that the problem to decide the regularity of \hpc $\ccH$
for given regular languages $L_1$ and $L_2$ is \NL-complete. 
In particular it can be solved efficiently in parallel with Boolean circuits 
of polynomial size and poly-log depth, because
\NL is contained in \emph{Nick's Class} $\textbf{NC}_2$ (see e.g. \cite[Thm.~16.1]{pap94}). 

Our  \NL-result is based on the fact that the input is 
given by DFAs accepting $L_1$ and $\ov{L_2}$. 
It is open, what happens if the input is given in a more concise form, say
the input is given by NFAs accepting $L_1$ and $L_2$ (or $\ov{L_2}$).

Another result of \cite{DieKop11} says that the time complexity of the same
problem is in $\Oh(n^8)$. The full proof of this fact is quite involved, and it employs different ideas.   
It will appear elsewhere. It is open whether the $\Oh(n^8)$
time bound
is optimal. A further improvement on this time bound seems however  to ask for quite different ideas. So far, the best algorithm known (to us) considers all 
pairs of states in the automaton $\cA$. There are $\Omega(n^8)$ pairs and
it is unclear how to avoid this bound. 

There is also a very natural  variant of \hpc which was introduced in 
\cite{DiekertKM09}. It has been called \emph{partial hairpin completion} 
and further investigated in \cite{ManeaMM10}, where the operation has been called \emph{hairpin lengthening}.
The partial hairpin completion of $L_1$ and $L_2$ is given  by the set of words $\gam \alp \bet \ov \alp \ov {\gam'}$, where $\gam'$ is a prefix $\gam$ and $\gam \alp \bet \ov \alp \in L_1$ or $\gam$  is a prefix $\gam'$  and 
$\alp \bet \ov \alp \ov {\gam'}\in L_2$.

Again, the partial \hpc of a regular language is linear context-free, 
but not regular, in general. It is open whether regularity of the partial \hpc of regular languages is decidable.

%
\newcommand{\Ju}{Ju}\newcommand{\Ph}{Ph}\newcommand{\Th}{Th}\newcommand{\Ch}{C%
h}\newcommand{\Yu}{Yu}\newcommand{\Zh}{Zh}

\end{document}